\documentclass[11pt]{article}
\usepackage{amsfonts,amsmath,amsthm,amssymb,fullpage}
\usepackage[T1]{fontenc}
\usepackage[colorlinks,citecolor=blue,linkcolor=blue,urlcolor=black]{hyperref}
\usepackage{comment, enumerate}
\usepackage{soul,color}
\usepackage{mathrsfs}
\usepackage{xspace}
\usepackage{mathdots}
\usepackage{algorithm}
\usepackage{algpseudocode}
\usepackage{bbm}
\usepackage{tikz}
\usepackage{url}
\usepackage{tabu}
\usepackage{xfrac}
\usepackage[normalem]{ulem}
\usepackage{mathtools}
\usepackage{etextools}
\usepackage{ifthen}

\algnewcommand{\algorithmicand}{\textbf{AND }}
\algnewcommand{\algorithmicor}{\textbf{OR }}
\algnewcommand{\algorithmicxor}{\textbf{XOR }}
\algnewcommand{\algorithmicnot}{\textbf{NOT }}
\algnewcommand{\OR}{\algorithmicor}
\algnewcommand{\AND}{\algorithmicand}
\algnewcommand{\XOR}{\algorithmicxor}
\algnewcommand{\NOT}{\algorithmicnot}
\algnewcommand{\var}{\texttt}

\algnewcommand{\algorithmicbreak}{\textbf{break}}
\algnewcommand{\Break}{\algorithmicbreak}

\usetikzlibrary{hobby}
\usetikzlibrary{decorations.markings}
\usetikzlibrary{quotes,angles}

\DeclarePairedDelimiter\abs{\lvert}{\rvert}%
\DeclarePairedDelimiter\norm{\lVert}{\rVert}%

\makeatletter
\let\oldabs\abs
\def\abs{\@ifstar{\oldabs}{\oldabs*}}
\let\oldnorm\norm
\def\norm{\@ifstar{\oldnorm}{\oldnorm*}}
\makeatother

\theoremstyle{definition} 

\newtheorem{theorem}{Theorem}[section]
\newtheorem{lemma}[theorem]{Lemma}

\newtheorem{cor}[theorem]{Corollary}

\theoremstyle{definition}
\newtheorem{definition}[theorem]{Definition}

\newenvironment{proofof}[1]{\begin{proof}[Proof of #1]}{\end{proof}}

\def\ShowAuthNotes{1}
\ifnum\ShowAuthNotes=1
\newcommand{\authnote}[2]{\ \\ \textcolor{red}{\parbox{0.9\linewidth}{[{\footnotesize {\bf #1:} { {#2}}}]}}\newline}
\else
\newcommand{\authnote}[2]{}
\fi

\renewcommand{\epsilon}{\varepsilon}

\renewcommand{\Pr}{\operatorname*{\mathbf{Pr}}}

\newcommand{\poly}{\operatorname{\mathrm{poly}}}
\newcommand{\polylog}{\poly\log}

\renewcommand{\tilde}{\widetilde}

\newcommand{\sol}{\mathsf{sol}}
\newcommand{\val}{\mathsf{val}}
\newcommand{\supp}{\mathsf{supp}}

\newcommand{\defproblemu}[3]{
  \vspace{2mm}
  \vspace{1mm}
\noindent\fbox{
  \begin{minipage}{0.95\textwidth}
  {\bf{#1}} \\
  {\bf{Input:}} #2  \\
  {\bf{Task:}} #3
  \end{minipage}
  }
  \vspace{2mm}
}

\usepackage{caption}
\usepackage{transparent}
\usepackage{subcaption}

\usetikzlibrary{shapes.geometric}
\usetikzlibrary{arrows}

\tikzset{
triangle/.style={
  draw,solid,isosceles triangle,shape border rotate=90},
}

\title{On Problems Related to Unbounded SubsetSum: A Unified Combinatorial Approach}

\author{
  Mingyang Deng \\
  Massachusetts Institute of Technology \\
  \texttt{dengm@mit.edu} \\
  \and
  Xiao Mao \\
  Massachusetts Institute of Technology \\
  \texttt{xiao\_mao@mit.edu} \\
  \and
  Ziqian Zhong \\
  Massachusetts Institute of Technology \\
  \texttt{ziqianz@mit.edu} \\
}

\date{}

\begin{document}

\maketitle

\begin{abstract}




Unbounded SubsetSum is a classical textbook problem: given integers $w_1,w_2,\cdots,w_n\in [1,u],~c,u$, we need to find if there exists $m_1,m_2,\cdots,m_n\in \mathbb{N}$ satisfying $c=\sum_{i=1}^n w_im_i$. In its all-target version, $t\in \mathbb{Z}_+$ is given and answer for all integers $c\in[0,t]$ is required. In this paper, we study three generalizations of this simple problem: All-Target Unbounded Knapsack, All-Target CoinChange and Residue Table. By new combinatorial insights into the structures of solutions, we present a novel two-phase approach for such problems. As a result, we present the first near-linear algorithms for CoinChange and Residue Table, which runs in $\tilde{O}(u+t)$ and $\tilde{O}(u)$ time deterministically. We also show if we can compute $(\min,+)$ convolution for $n$-length arrays in $T(n)$ time, then All-Target Unbounded Knapsack can be solved in $\tilde{O}(T(u)+t)$ time, thus establishing sub-quadratic equivalence between All-Target Unbounded Knapsack and $(\min,+)$ convolution.
\end{abstract}

\section{Introduction}
\label{sec:intro}

\subsection{Background}

Consider the following problem, \textit{All-Target Unbounded SubsetSum}. Given $w_1,w_2,\cdots,w_n,t\in \mathbb{Z}_+$, for each $c\in [0,t]$ we want to find if there are some $w$'s with sum $t$, where every $w$ could be used multiple times. More formally, we want to find if there exists $m_1,m_2,\cdots,m_n \in \mathbb{N}$ satisfying $c=\sum_{i=1}^n w_im_i$, for each $c\in [0,t]$. We call a $c$ for which such a $m$ exists \textit{feasible} and \textit{non-feasible} otherwise.

While this problem is relatively simple and could be easily solved in $\tilde{O}(n+t)$\footnote{$\tilde{O}$ hides polylogarithmic factors.} time by repeated convolutions, many of its generalizations are not well-understood. In this paper, we address three related problems that have been studied separately, \textit{All-Target Unbounded Knapsack}, \textit{All-Target CoinChange} and \textit{Residue Table}.

In All-Target CoinChange, for each $c$ one needs to find the minimum possible $\sum_{i=1}^n m_i$ while satisfying $c=\sum_{i=1}^n w_im_i$. Intuitively, $w$'s are the possible values of the coins and the cashier needs to find the minimum number of coins with values summing up to $c$. In All-Target Unbounded Knapsack, each $w_i$ is associated with an integer $p_i$, and one needs to find the maximum possible $\sum_{i=1}^n p_im_i$ while satisfying $c=\sum_{i=1}^n w_im_i$. Considering $(w_i,p_i)$ as a type of item with weight $w_i$ and profit $p_i$, we are trying to find the maximum profit for items with total weight $c$. In their corresponding Single-Target version, only answer for one target $c=t$ is required.

CoinChange and Unbounded Knapsack are two textbook problems for dynamic programming. While Single-Target CoinChange can be solved in $\tilde{O}(t)$ time with convolution and repeated squaring \cite{chan2020change}, the best known algorithm for All-Target CoinChange has long been a $\tilde{O}(t^{3/2})$-time algorithm \cite{lincoln2020monochromatic}, until the recent improvement to $\tilde{O}(t^{4/3})$ by Chan and He \cite{chan2020more}. They also presented a $O(u^2\log(u)+t)$ time algorithm, which is more efficient when $u\ll t$.

On Unbounded Knapsack, Cygan et al. \cite{cygan2019problems} showed a sub-quadratic algorithm ($O(t^{2-\epsilon})$ for some $\epsilon>0$) for Single-Target Unbounded Knapsack would imply a sub-quadratic algorithm for $(\min,+)$ convolution. Axiotis and Tzamos \cite{axiotis2018capacitated} showed if $n$-length $(\min,+)$ convolution can be solved in $T(n)$ time, then Single-Target Unbounded Knapsack can be solved in $O(T(u)\polylog(t))$ time, thereby establishing a sub-quadratic equivalence between Single-Target Unbounded Knapsack and $(\min,+)$ convolution. However, their method does not apply for the All-Target version. Chan and He \cite{chan2020more} recently presented a $O(u^2\log u+t)$-time algorithm for All-Target Unbounded Knapsack.

In Residue Table, for each $t\in [0,w_1)$, we need to find the smallest $\sum_{i=1}^n w_ic_i$ among all $c_1,c_2,\cdots,c_n \in \mathbb{N}$ satisfying $t\equiv \sum_{i=1}^n w_ic_i\pmod {w_1}$. That is, we need to compute the smallest feasible sum with remainder $t$ modulo $w_1$. Residue Table is first introduced by Brauer and Shockley \cite{shockley1962problem} to tackle the Frobenius problem. This table would allow one to check in $O(1)$ time if a sum is feasible for Unbounded SubsetSum, by comparing it with the minimum feasible sum with the same remainder modulo $w_1$, since we can always add more $w_1$'s to a feasible sum to get another feasible sum. Klein \cite{klein2021fine} presented an algorithm computing the table in $\tilde{O}(u^{3/2})$ time.








\subsection{Main results}

In this paper, we present new insights on structures of solutions to these two problems. Crucial to our observations is focusing only on optimal-valued solutions with minimal lexical order and the optimal substructure property of these solutions. The optimal substructure property of the solutions enables us to ``peel'' solutions, removing duplicated items to arrive at solutions without duplicated items, which we call \textit{kernel}s. From kernels, we can ``propagate'' backwards, adding duplicated items for item types in kernels, to get optimal solutions. Therefore, we can tackle these problems with a two-phase approach: compute the solutions for the kernels and propagate.

With this approach, we arrive at new results for the three problems:

\begin{description}
    \item[All-Target Unbounded Knapsack (Theorem \ref{theo:knapsack})] Let $T(n)$ be the time required for $n$-length $(\min,+)$ convolution, All-Target Unbounded Knapsack can be solved in $\tilde{O}(T(u)+t)$ time.
    \item [All-Target CoinChange (Theorem \ref{theo:coin})] All-Target CoinChange can be solved in $\tilde{O}(u+t)$ time.
    \item [Residue Table (Theorem \ref{theo:restable})] Residue Table can be computed in $\tilde{O}(u)$ time.
\end{description}

Our algorithms are relatively simple and practical. Notice that if we can solve All-Target Unbounded Knapsack in $T(u+t)$ time, we can compute $(\min,+)$ convolution for $n$-length arrays in $O(T(n))$ time: to compute $(\min,+)$ convolution of $a_1,a_2,\cdots,a_n$ and $b_1,b_2,\cdots,b_n$, let $x$ be a sufficiently large integer, create items $(4n+i,a_i+x)$ and $(2n+i,b_i)$ for every $i$ and run All-Target Unbounded Knapsack. The optimal value for $6n+s$ will be value at position $s$ in the convolution result plus $x$. Therefore our purposed results are all optimal, up to log factors.

\subsection{Technical overview}

\paragraph{Structural property under lexical Order}
Our algorithm sprouts from the recent observation made by Klein \cite{klein2021fine} which implies that for unbounded knapsack, it suffices to not use too many types of items. Specifically, if we fix an arbitrarily chosen lexical order, then the lexicographically smallest optimal solution $\sol(j)$ for each feasible target sum $j \in [1, t]$ has a support of logarithmic size (i.e. $\abs{\supp(\sol(j))} = O(\log t)$). Since CoinChange can be viewed as a special case of unbounded knapsack, this structural property also applies. 

\paragraph{Witness propagation using the optimal substructure property}
The most essential technique in our paper is ``witness propagation.'' It exploits the following \emph{optimal substructure property} for lexicographically smallest optimal solutions: for any target $j$, and for any ``witness'' $x \in \supp(\sol(\tau))$, $\sol(j - w_x)$ is equal to $\sol(j)$ with the multiplicity of $x$ decreased by $1$. Suppose that we can somehow compute $\sol(j)$ for every feasible target $j \in [1, t]$ whose optimal solution is of logarithmic size (i.e. $\abs{\sol(j)} = O(\log u)$), defined as a ``kernel,'' then since the total number of ``witnesses'' is logarithmic, we can propagate the solutions forward by enumerating witnesses and finds the optimal solution for all other feasible targets in $[1, t]$. This \emph{witness propagation} runs in $\tilde O(t)$ time.

\paragraph{Min-witness with arbitrary lexical order}
For unbounded knapsack, since the kernels are in $[1, O(u \log u)]$, the optimal solutions for kernels can be computed in $O(T(O(u \log u))\log{u}) = O(u ^ 2 / 2 ^ {\Omega(\sqrt{\log u})})$ time by repeating $(\min,+)$ convolution on this interval for $O(\log u)$ times. For CoinChange, although it is easy to compute the size of the optimal solutions for kernels in $\tilde O(u)$ time using FFT, it is not easy to find the lexicographically smallest optimal solutions as they seem to require finding ``minimum witnesses'' for convolution, for which currently the best algorithm only runs in $\tilde O(u ^ {1.5})$ time (e.g. \cite{lingas2018extreme}). However, since the lexical order can be arbitrary, we can overcome this barrier by picking certain orders. We purpose two different approaches of independent interest. First, we first show minimum witness is easy to compute under a random order. Also, we provide a deterministic construction which computes $\sigma$ and minimum witnesses alongside in similar spirit.





\section{Preliminaries}
\label{sec:prelim}


We first formally define the three problems.

For a set of integers $w_1,w_2,\cdots,w_n$, we call a sum $c$ \textit{feasible} if $c = \sum_{i=1}^n w_im_i$ where $m_i\in \mathbb{N}$. We call such $(m_1,m_2,\cdots,m_n)$ a \textit{solution} to sum $c$. The \emph{support} of $m$ is defined to be $\supp(m) = \{j \mid m_j > 0\}$. The \emph{size} of $m$ is defined to be $\abs{m} = \sum_{i = 1} ^ {n}{m_i}$.

\defproblemu{All-Target Unbounded Knapsack}
{$n,u\in \mathbb{Z}_+$, $w_1,w_2,\cdots,w_n\in [1,u]\cap \mathbb{Z}$, $p_1,p_2,\cdots,p_n\in \mathbb{Z}$}
{Define value of solution $\val(m)=\sum_{i=1}^n p_im_i$. For each integer $c\in [1,t]$, output maximum possible value of a solution to sum $c$, or $-\infty$ if not feasible.}

Notice that in our definition we require sum of weights to be exactly $c$ instead of not exceeding $c$ in some other definitions.

\defproblemu{All-Target CoinChange}
{$n,u\in \mathbb{Z}_+$, $w_1,w_2,\cdots,w_n\in [1,u]\cap \mathbb{Z}$}
{Define value of solution $\val(m)=-\sum_{i=1}^n m_i$. For each integer $c\in [1,t]$, output maximum possible value of a solution to sum $c$, or $-\infty$ if not feasible.}

With this set of notation, it's clear that All-Target CoinChange is a special case of All-Target Unbounded Knapsack.

\defproblemu{Residue Table}
{$u\in \mathbb{Z}_+$, $w_1,w_2,\cdots,w_n\in [1,u]\cap \mathbb{Z}$}
{For each integer $c\in [0,w_1-1]$, output minimum feasible $s\ge 0$ where $s\equiv c\pmod {w_1}$.}

We call a solution for All-Target Unbounded Knapsack and All-Target CoinChange \textit{optimal} iff it is of maximal possible value for the same sum. We consider any solution of Residue Table optimal.

A \emph{lexical order} $\sigma = (\sigma_1, \sigma_2, \cdots, \sigma_n)$ is a permutation of $[1, n]$ denoting an order between the items. Items that appear earlier in $\sigma$ are lexicographically smaller. Solution $A = (a_1, a_2, \cdots a_n)$ is lexicalgraphically smaller than $B = (b_1, b_2, \cdots, b_n)$ if there exists a $j \in [1, n]$ such that for all $k < j$, $a_{\sigma_k} = b_{\sigma_k}$ and $a_{\sigma_j} > b_{\sigma_j}$. We denote this by $A < B$.

The lexicographicall smallest optimal solution for sum $j$ under $\sigma$ is denoted by $\sol(j, \sigma)$, and $\sol(j,\sigma)=\emptyset$ for unfeasible $j$'s. A feasible target $j \in [0, u]$ is called an $x$\emph{-kernel} under $\sigma$ if $\abs{\sol(j, \sigma)} \le x$. Let $\sol_m(u,\sigma)$ be $\sol(j,\sigma)$ for the minimum feasible $j$ with remainder $u$ modulo $w_1$, or $\emptyset$ if such $j$ does not exist.

Due to the additive nature, our problem is closely related to convolutions. We define boolean convolutions and $(\min,+)$ convolutions.

\begin{definition} [boolean convolution]
    Define arrays of $\{0,1\}$ boolean arrays. Given two boolean arrays $a[0 \cdots n]$ and $b[0 \cdots m]$, define their boolean convolution as $c[0 \cdots n + m]$ where $c[i]=\lor_{j+k=i} (a[j]\land b[k])$.
\end{definition}

Boolean convolution can be computed in $O((n+m)\log(n+m))$ time by regular convolution via Fast Fourier Transform (e.g. \cite{furer2014fast}).

\begin{definition} [$(\min,+)$ convolution]
    Given two arrays, $a[0 \cdots n]$ and $b[0 \cdots m]$, define their convolution as $c[0 \cdots n + m]$ where $c[i]=\min_{j+k=i} (a[j]+b[k])$.
\end{definition}
\begin{lemma} [\cite{williams2018faster}] \label{lemma:MinPlusConv}
    $(\min,+)$ convolution be computed in $O((n + m)^2/2^{\Omega(\sqrt{\log(n+m)})})$ time.
\end{lemma}

\section{Combinatorial Properties}

We start by introducing Lemma 1 in \cite{klein2021fine}, which implies solutions have logarithmic sized-support in Unbounded SubsetSum.

\begin{lemma} [Lemma 1 in \cite{klein2021fine}] \label{lemma:capss}
    For Unbounded SubsetSum, for any lexical order $\sigma$ and any feasible target $j$, let $x = \sol(j, \sigma)$ be the lexicographically smallest solution for $j$ under $\sigma$, then $\prod_{i\ne\sigma_1}{(x_i + 1)} \le w_{\sigma_1}\le u$.
\end{lemma}

\begin{cor} \label{cor:logss}
For Unbounded SubsetSum, for any lexical order $\sigma$ and any feasible target $j$, $\abs{\supp(\sol(j, \sigma))}\le \log_2{u}+1$. Thus for residue table, the support sizes of the solutions are also $\le \log_2{u}+1$.
\end{cor}

\begin{proof}
    Let $x=\sol(j,\sigma)$, by Lemma \ref{lemma:capss} $2^{\abs{\supp(x)}-1}\le \prod_{i\ne\sigma_1}{(x_i + 1)} \le j+1$, so $\abs{\supp(x)}\le \log_2{u}+1$.
\end{proof}

We can extend the lemma to the valued version with a similar adjusting argument.

\begin{lemma} \label{lemma:cap}
    For Unbounded Knapsack, for any lexical order $\sigma$ and any feasible target $j$, let $x = \sol(j, \sigma)$, then $\prod_{i=1}^n{(x_i + 1)} \le u+1$.
\end{lemma}

\begin{proof}
Suppose otherwise, consider all integer sequences $(y_1,y_2,\cdots,y_n)$ so that $y_i\in [0,x_i]$. Notice $0\le \sum_{i=1}^n y_iw_i\le \sum_{i=1}^n x_iw_i=j$ and the number of $y$'s is $\prod_{i=1}^n{(x_i + 1)}>j+1$, by pigeonhole principle there is $y\ne y'$ so that $\sum_{i=1}^n y_iw_i=\sum_{i=1}^n y'_iw_i$. $j=\sum_{i=1}^n (x_i-y_i+y'_i)w_i=\sum_{i=1}^n (x_i+y_i-y'_i)w_i$. If $\sum_{i=1}^n y_ip_i\ne \sum_{i=1}^n y'_ip_i$, one of $\{x_i-y_i+y'_i\}_i$ and $\{x_i+y_i-y'_i\}_i$ is a solution for $j$ with larger value. Otherwise, both of them have value equal to $x$'s, and one of them will be lexicographically smaller than $x$. In both cases we get a contradiction with optimality of $x$.
\end{proof}

\begin{cor} \label{cor:logknap}
For Unbounded Knapsack and CoinChange, for any lexical order $\sigma$ and any feasible target $j$, $\abs{\supp(\sol(j, \sigma))}\le \log_2(j+1)$.
\end{cor}

\begin{proof}
    Let $x=\sol(j,\sigma)$, by Lemma \ref{lemma:cap} $2^{\abs{\supp(x)}}\le \prod_{i=1}^n{(x_i + 1)} \le j+1$, so $\abs{\supp(x)}\le \log_2(j+1)$.
\end{proof}

Corollary \ref{cor:logknap} only implies $O(\log(t))$-size supports instead of $O(\log(u))$ as in Corollary \ref{cor:logss}, but we can use the following lemma.

\begin{lemma} [Lemma 4.1 in \cite{chan2020more}\footnote{The original proof is for unweighted case, but it can be easily modified to prove the weighted case.}] \label{lemma:capqr}
    Let $s$ be any type with maximum value/weight ratio, for any feasible target $j>u^2$, $j-s$ must also be feasible, and an optimal solution for $j$ might be found by adding an item of type $s$ to any optimal solution for $j-s$. Thus if $t>u^2$, we can first compute optimal solutions for $[0,u^2]$, and then for each $j\in (u^2,t]$, simply add item $s$ to a solution for $j-s$ to get a solution for $j$.
\end{lemma}

By investigating further into these lexicographically minimal solutions, we can find the following structural property between solutions, \textit{optimal substructure property}.

\begin{lemma} [optimal substructure property] \label{lemma:substructure}
    For any lexical order $\sigma$, a feasible target $j \in [1, t]$ and a ``witness'' $x \in \supp(\sol(j, \sigma))$, let $\sol(j, \sigma) = (u_1, u_2, \cdots, u_n)$. Define $(v_1, v_2, \cdots, v_n)$ as follows:
    \begin{align*}
        v_k = 
        \begin{cases}
            u_k                                 & \text{if } k \ne w \\
            u_k - 1                              & \text{if } k = w
        \end{cases}.
    \end{align*} 
    we have $\sol(j - w_x, \sigma) = v$.
\end{lemma}
\begin{proof}
    Firstly, we prove $\val(\sol(j - w_x, \sigma)) = \val(v)$ by showing that neither $\val(\sol(j - w_x, \sigma)) < \val(v)$ nor $\val(v) < \val(\sol(j - w_x, \sigma))$ can hold. If $\val(\sol(j - w_x, \sigma)) < \val(v)$, then adding item $x$ to $\sol(j - w_x, \sigma)$ gives a better solution for $j$, which is impossible. If $\val(v) < \val(\sol(j - w_x, \sigma)$, then removing $x$ from $\sol(j, \sigma)$ gives a better solution for $j - w_x$, which is impossible.

    We then argue that neither $\sol(j - w_x, \sigma) < v$ nor $v < \sol(j - w_x, \sigma)$ can hold. If $\sol(j - w_x, \sigma) < v$, then adding item $x$ to $\sol(j - w_x, \sigma)$ gives a lexicographically smaller solution for $j$ with the same value, which is impossible. If $v < \sol(j - w_x, \sigma)$, then removing $x$ from $\sol(j, \sigma)$ gives a lexicographically smaller solution for $j - w_x$ with the same value, which is impossible.
\end{proof}

We also have the following modular analog, which can be proved similarly.

\begin{lemma} [optimal substructure property, modular] \label{lemma:substructurem}
    For any lexical order $\sigma$, . For any $j$ and a ``witness'' $x \in \supp(\sol_m(j, \sigma))$, let $\sol_m(j, \sigma) = (u_1, u_2, \cdots, u_n)$. Define $(v_1, v_2, \cdots, v_n)$ as follows:
    \begin{align*}
        v_k = 
        \begin{cases}
            u_k                                 & \text{if } k \ne w \\
            u_k - 1                              & \text{if } k = w
        \end{cases}.
    \end{align*} 
    we have $\sol_m((j - w_x)\bmod w_1, \sigma) = v$.
\end{lemma}

\section{Witness Propagation}

With the help of optimal substructure property, we introduce the idea of witness propagation.

By Corollary \ref{cor:logss} or \ref{cor:logknap}, we have support of the solutions are logarithmic-sized. Let $k$ be an upper bound of the size of supports. Suppose we have an array of solutions $sol[0,\cdots,t]$ where $sol[j]=\sol(j,\sigma)$ for all $k$-kernel $j$ and $sol[j]$ is a valid solution or $\emptyset$ for other $j$'s. The idea is to gradually propagate from existing solutions, each time adding one more item in some existing solution. By the optimal substructure property, every optimal solution can be thus found from kernel formed by its support. We give the following algorithm \ref{alg:propagation}. 

\begin{algorithm} [H]
    \caption{Witness propogation} \label{alg:propagation}
    \begin{algorithmic}[1]
        \Procedure{Propagation}{}
            \For {$j \in [1, t]$} \label{line:candj}
                \If {$sol[j]\ne \emptyset$}
                    \For {$x \in \supp(sol[j])$}
                        \State $s \gets sol[j]$
                        \State $s_x \gets s_x + 1$
                        \If {($sol[j + x]=\emptyset$) \OR ($\val(s) > \val(sol[j+x])$) \par\hspace{3.8em} \OR (($\val(s) = \val(sol[j+x])$) \AND ($s < sol[j+x]$))}
                            \State $sol[j+x] \gets s$ \label{line:update}
                        \EndIf
                    \EndFor
                \EndIf
            \EndFor
        \EndProcedure
    \end{algorithmic}
\end{algorithm}

\begin{lemma}
    Given $sol[j]=\sol(j,\sigma)$ for all $k$-kernels $j$ where $k$ is an upper bound of support sizes, algorithm \ref{alg:propagation} decides whether each $j \in [1, t]$ is feasible, and correctly computes an optimal solution for every feasible $j$.
\end{lemma}
\begin{proof}
    It suffices to show that for each feasible $j ^ {\prime}$, $sol[j']=\sol(j ^ {\prime}, \sigma)$ is once examined on line \ref{line:update}.
    
    We prove by induction on $\abs{\sol(j ^ {\prime}, \sigma)}$. Firstly, we have assumed that $sol[j']=\sol(j ^ {\prime}, \sigma)$ for all $j ^ {\prime}$ where $\abs{\sol(j ^ {\prime}, \sigma)} \le 2\log_2(u)+1$, which are the $(2\log_2(u)+1)$-kernels.
    
    Suppose $\abs{\sol( j^ {\prime}, \sigma)} > 2\log_2(u)+1$. From Lemma \ref{lemma:cap} we know $\supp(\sol(j',\sigma))\le \log_2(j+1)\le 2\log_2(u)+1$, so there exists an $x ^ {\prime} \in \supp(\sol(j ^ {\prime}, \sigma))$ with multiplicity at least 2, then from Lemma \ref{lemma:substructure} we know that $x ^ {\prime} \in \supp(\sol(j ^ {\prime} - w_x, \sigma))$. By induction hypothesis $sol[j'-w_x]=\supp(\sol(j ^ {\prime} - w_x, \sigma))$, therefore $\sol(j ^ {\prime}, \sigma)$ will be examined on line \ref{line:update} when $j = j ^ {\prime} - w_x$ and $x = x ^ {\prime}$.
\end{proof}

\begin{lemma} \label{lemma:propagate}
   For any lexical order $\sigma$, given $\sol(j,\sigma)$ for every $(2\log_2 u+1)$-kernel $j$, All-Target Knapsack can be solved in $O(t\log^2(u))$ time.
\end{lemma}
\begin{proof}
    We can compute answers for $u^2+1,u^2+2,\cdots,t$ by Lemma \ref{lemma:capqr}, so we may assume $t\le u^2$. $k=2\log_2 u+1\ge \log_2(t+1)$ would be an upper bound on support sizes by Corollary \ref{cor:logknap}.

    We implement algorithm \ref{alg:propagation} by storing each $sol$ as an array which size equals to its support, recording non-zero positions and corresponding values in lexical order. Comparisons can then be done in time linear to array sizes.
    
    Size of supports of the internal $sol$'s should be no larger than $k+1\le 2\log_2(u)+2$ by the nature of this algorithm. For every feasible $j$, there are $O(\log(u))$ witnesses $x$ on line \ref{line:candj} and updating for each witness takes time $O(\log(u))$. Therefore the total time complexity for propagation should be $O(t\log^2(u))$.
\end{proof}

By modifying Algorithm \ref{alg:propagation} in a modular fashion, we can prove a similar result for Residue Table.

\begin{lemma} \label{lemma:propagatemod}
   For any lexical order $\sigma$, given $\sol(j,\sigma)$ for every $\log_2(u+1)$-kernel $j$, residue table can be computed in $O(u\log^2(u))$ time.
\end{lemma}
\begin{proof}
    We update $sol_m[j\bmod w_1]$ with all $\sol(j,\sigma)$'s and propagate on $sol_m$ similar to Algorithm \ref{alg:propagation}, in a Dijkstra-like fashion. Instead of looping through $j$ in increasing order, we iterate through $j$'s in the order of solution sizes. We maintain a priority queue with Fibonacci heap, each time popping the entry with minimum solution size and propagating with it. When we propagate, we decrease key in the Fibonacci heap in $O(1)$ time. The correctness can be proved by an induction on solution size.
\end{proof}

\section{Kernel Computation}

With Lemma \ref{lemma:propagate} and \ref{lemma:propagatemod}, we only need to consider the computation of $O(\log u)$-kernels. Intuitively, we could set up an array of values for each weight and convolve it with itself $O(\log u)$ times, but while we can get an optimal-valued solution in this way, it's not necessarily of minimal lexical order. However, if we can for the convolutions, find the minimum witness with respect to the lexical order, we can compute the solutions we want. We illustrate the idea in algorithm \ref{alg:kernelcompidea} for Unbounded Knapsack. For Coinchange and Residue Table, simply modify $v,f$ to be booleans and change $(\max,+)$ convolution to boolean convolution.

\begin{algorithm} [H]
    \caption{Kernel Computation with Minimum Witness} \label{alg:kernelcompidea}
    \begin{algorithmic}[1]
        \Procedure{Kernel Computation($\sigma$)}{}
            \State $k\gets \lfloor 2\log_2(u)+1\rfloor$
            \State Initialize $v[0,\cdots,ku]$ and $f[0,\cdots,u]$ to be $-\infty$, $sol[0,\cdots,ku]$ to be $\emptyset$
            \State $v[0]\gets 0$
            \State $sol[0]\gets\{\}$
            \For {$i \in [1, n]$}
                \State $f[w_{\sigma_i}]\gets p_{\sigma_i}$
            \EndFor
            \For {$j \in [1, k]$}
                \State Compute $(\max,+)$ convolution of $v$ and $f$ and store in $v'$
                \State For each $v'[i]$, find the minimum $t[i]$ (``witness'') so that $f[w_{\sigma_{t[i]}}]$ contributed to $v'[i]$
                \For {$i\in [1,ku]$}
                    \If {$v'[i]\ne -\infty$}
                        \State $v[i]\gets v'[i]$
                        \State $c\gets \sigma_{t[i]}$
                        \State $s\gets sol[i-w_c]$
                        \State $s_c\gets s_c+1$
                        \State $sol[i]\gets s$
                    \EndIf
                \EndFor
            \EndFor
        \EndProcedure
    \end{algorithmic}
\end{algorithm}

\begin{lemma} \label{lemma:kernelcompidea}
   Algorithm \ref{alg:kernelcompidea} correctly computes $\sol$ with respect to $\sigma$ for all $(2\log_2(u)+1)$-kernels.
\end{lemma}
\begin{proof}
    Let $k=\lfloor 2\log_2(u)+1 \rfloor$. Clearly the $k$ kernels are in $[0,ku]$ since they are sum of at most $k$ $w$'s. We perform induction on $j$ in the code: after running the inner loop for $j$ times, correct $\sol$ has been computed for all $j$-kernels. Except $0$, all $j$-kernel can be resulted from adding one element to a $j-1$-kernel, and the optimal value is computed by the convolution. To minimize the lexical order, we find the smallest possible witness, which is the smallest possible starting element (smallest $j$ so that $\sigma_j$ position could be non-empty). The remaining part is also minimal possible by induction hypothesis.
\end{proof}

\subsection{Minimum Witness for $(\max,+)$ Convolution}

For $(\max,+)$ convolution, we can find the minimum witness during the convolution by letting $v_w=(n+1)v$, $f_w[w_{\sigma_i}]=(n+1)f[w_{\sigma_i}]-i$, and we can tell the minimum witness by the remainder modulo $n+1$.

\begin{theorem} \label{theo:knapsack}
Let $T(n)$ be the time required for $n$-length $(\min,+)$ convolution in $T(n)$ time. All-Target Unbounded Knapsack can be solved in $O(T(u)\log^3(u)+t\log^2(u))$ time.
\end{theorem}
\begin{proof}
    From Lemma \ref{lemma:propagate} it suffices to compute $\sol(j, \sigma)$ for all $(2\log_2{u}+1)$-kernels. We modify algorithm \ref{alg:kernelcompidea} to compute minimum witness: let $v_w[i]=(n+1)v[i]$ and $f_w[\sigma_i]=(n+1)f[\sigma_i]-i$ and compute $(\max,+)$ convolution on $v_w$ and $f_w$ in the inner loop. $v'$ and $t$ can both be computed from the convolution result.
    
    The algorithm calculates $(\max,+)$ convolution on length-$O(u\log(u))$ arrays $O(\log(u))$ times, taking $O(T(u)\log^3 (u))$ time. Combining with Lemma \ref{lemma:propagate}, the final time complexity would be $O(T(u)\log^3(u)+t\log^2(u))$.
\end{proof}

\subsection{Minimum Witness for Boolean Convolution under Random Order}

While computing minimum witness for boolean convolution is hard and the current best algorithm runs in $\tilde{O}(u^{1.5})$ time (e.g. \cite{lingas2018extreme}), we can overcome this barrier by carefully picking $\sigma$. In this subsection, we show minimum witness is likely easier for a randomly chosen $\sigma$.

\begin{theorem} [Minimum witness finding for random permutations] \label{thm:minwitness}
    Given boolean arrays $a[0 \cdots n-1]$ and $b[0 \cdots m-1]$, and their convolution $c[0 \cdots n + m-2]$. For every index $i$ where $c[i] > 0$, let $x[i] = \{k \mid a[k]=b[i - k]=1\}$.
    
    For a uniformly randomly chosen lexical order $\sigma$ over all $n!$ possible permutations, in expected $\tilde O(n+m)$ time we can compute an array $d[0 \cdots n + m-1]$ where for each index where $c[i] > 0$, $d[i]$ is equal to the element in $x[i]$ that is smallest under $\sigma$.
\end{theorem}

We purpose Algorithm \ref{alg:witness}, which computes $d[0 \cdots n + m -2]$ given $\sigma$, $a[0 \cdots n-1]$, $b[0 \cdots m-1]$ and $c[0 \cdots n + m-2]$. Without loss of generality we assume that $n$ is a power of two.

\algdef{SE}[DOWHILE]{Do}{doWhile}{\algorithmicdo}[1]{\algorithmicwhile\ #1}%

\begin{algorithm} [H]
    \caption{Minimum witness finding under random order} \label{alg:witness}
    \begin{algorithmic}[1]
        \Procedure{MinimumWitnessFinding}{}
            \State $x \gets [\emptyset] ^ {n}$ 
            \State $d \gets [-1] ^ {n}$
            \For {$l = [1, 2, 4, \cdots n / 2, n]$}
                \State $a ^ {\prime} \gets [0] ^ {n}$
                \For {$i \in [0, l-1]$}
                    \State $a ^ {\prime}[\sigma_i] = a[\sigma_i]$
                \EndFor
                \State Compute the boolean convolution $c ^ {\prime}$ of $a ^ {\prime}$ and $b$.
                \While{there exists some $j\in[0,n+m-2]$ that $c'[j]>0$ and $d[j]=-1$}
                    \State Uniformly sample an array of witnesses $d ^ {\prime}$ from $a ^ {\prime}$, $b$ and $c ^ {\prime}$.\Comment{Lemma \ref{lemma:witness}}
                    \label{line:computewitness}
                    \For {$j \in [0, n + m - 2]$}
                        \If {$c ^ {\prime}[j] > 0$ \AND $d[j]=-1$}
                            \State $x[j] \gets x[j] \cup \{d ^ {\prime}[j]\}$
                            \If {$|x[j]|= c'[j]$}
                                \State $d[j]=\min_\sigma\{x[j]\}$ \Comment{Minimum with respect to order $\sigma$}
                            \EndIf
                            \label{line:addwitness}
                        \EndIf
                    \EndFor
                \EndWhile
            \EndFor
        \EndProcedure
    \end{algorithmic}
\end{algorithm}

In algorithm \ref{alg:witness}, we sample an array of witnesses with the following lemma.

\begin{lemma} [Witness sampling for boolean convolution (e.g. \cite{lingas2018extreme})] \label{lemma:witness}
    Given two boolean arrays $a[0 \cdots n-1]$ and $b[0 \cdots m-1]$, and their convolution $c[0 \cdots n + m-2]$. For each index $i$ where $c[i] > 0$, let $x[i] = \{j \mid a[j]=b[i - j]=1\}$.
    
    In $O((n + m) \polylog {(n + m)})$ expected time, we can compute an array $d[0 \cdots n + m-2]$ such that for each $i$ where $c[i] > 0$, $d[i]$ is equal to a uniformly random element from $x[i]$.
\end{lemma}

To prove the correctness of Algorithm \ref{alg:witness}, we first show the following lemma:
\begin{lemma} \label{lemma:logwitness}
    For each $j$ where $c[j] > 0$, with at least $1 - 1 / (n + m) ^ 5$ probability, there exists a power of two $l$ such that in algorithm \ref{alg:witness}, corresponding $c ^ {\prime}[j]\in (0,5\log_2(n+m))$.
\end{lemma}
\begin{proof}
    Consider another way to uniformly sample $\sigma$: we first partition indexes into first $n/2$ and the last $n/2$, then partition first $n/2$ further into two parts of $n/4$, and so on. After $\log_2 n$ such partitions, finally we permute the indexes within each part. It's clear that each intermediate part in this process corresponds to a prefix of $\sigma$ with a power of two length.
    
    With this process in mind, consider the smallest intermediate part with witness, we'll have $c'[j]>0$ for the corresponding $l$ and $c'[j]=0$ for $l/2$. These $c'[j]$ witnesses must all have been partitioned to indexes $[l/2,l)$ instead of $[0,l/2)$, and the probability that this happens is ${l/2\choose c'[j]}/{l\choose c'[j]}=\prod_{i=0}^{c'[j]-1}\frac{l/2-i}{l-i}\le 2^{-c'[j]}$. If $c'[j]>5\log_2(n+m)$, the probability is $<1/(n+m)^5$.
\end{proof}

\begin{proofof} {Theorem \ref{thm:minwitness}}
    We show algorithm \ref{alg:witness} suffice. For each power of two $l$, the algorithm samples a witness for each $j$, until all $c'[j]$ witnesses have been found for every $j$. The minimum witnesses are then computed from these witnesses and these $j$'s are no longer considered.
    
    For each $l$, let $u=\max_j\{c'[j]\}$ for $j$'s we're considering ($d[j]=-1$ in the algorithm). Clearly $u\le n$, and by Lemma \ref{lemma:logwitness} and union bound over all $j$'s $\Pr[u>5\log_2(n+m)]<1/(n+m)^4$, therefore $\text{E}[u]=O(\log(n+m))$.
    
    Consider some $j$, for every $10\log(n+m)c'[j]$ samplings, the probability that one witness is not found is $(1-1/c'[j])^{10\log(n+m)c'[j]}\le e^{-10\log(n+m)}=(n+m)^{-10}$. Therefore by union bound, after $10\log(n+m)u$ samplings, the probability that any of the $c'[j]$ witnesses is not found from some $j$ is $\le (n+m)^{-8}<1/2$, so each $10\log(n+m)u$ samplings give $\ge 1/2$ success rate and the expected number of samplings is $\le 20\log(n+m)E[u]=O(\log^2(n+m))$.
    
    Each of the sampling takes $\tilde{O}(n+m)$ expected time by Lemma \ref{lemma:witness} and we need to consider $\log_2 n$ $l$'s, so the algorithm takes $\tilde{O}(n+m)$ expected time.
\end{proofof}

This algorithm is of individual interest, and immediately gives near-optimal randomized algorithms solving All-Target Coinchange and Residue Table.

\subsection{Minimum Witness for Boolean Convolution with Adaptive Ordering} 

Under a deterministic setting, we can no longer sample a random $\sigma$. However, we can pick $\sigma$ and compute minimum witness together, in an adaptive fashion.

\begin{lemma} \label{lemma:minwitnessf}
    Given $a[0 \cdots n-1]$ and $b[0 \cdots m-1]$, and their convolution $c[0 \cdots n + m-2]$. For every index $i$ where $c[i] > 0$, let $x[i] = \{k \mid a[k]=b[i - k]=1\}$.
    
    For integer $k$, in deterministic $O(k(n + m) \polylog {(n + m)})$ time, we can compute $\min(|x[i]|,k)$ distinct members of $x[i]$, for each $i$ where $c[i]>0$.
\end{lemma}

The lemma may be proved in a similar fashion as in \cite{alon1996derandomization}. Alternatively, we can convert it into a instance of $k$-reconstruction problem defined in \cite{aumann2011finding} and use that algorithm.

\begin{lemma} [Deterministic hitting sets] \label{lemma:hittingset}
    Given sets $S_1,S_2,\cdots S_u$ where for every $i$, $S_i\subseteq \{1,2,\cdots,n\}$, $|S_i|\ge R$. We can compute a ``hitting set'' $S$ of size $\le (n/R)\log(u)$ deterministically so that $S\cap S_i\ne \emptyset$ for all $i$, in $\tilde{O}((n+u)R)$ time.
\end{lemma}

The lemma is a natural extension of Theorem 1 in \cite{aingworth1999fast} (which addresses the $n=u$ case) and can be proved in the same way.

\begin{theorem} \label{thm:adaptiveminwit}
Given $\tilde{O}(1)$ $n$-length boolean convolutions, in $\tilde{O}(n)$ time deterministically, we can compute a permutation $\sigma$, and the minimum witness of the convolutions with respect to $\sigma$.
\end{theorem}
\begin{proof}
    Assuming $n=2^\alpha$, we determine $\sigma$ and compute minimum witnesses in the following fashion: for $i$ from $\alpha$ down to $1$, determine first $2^{i-1}$ elements of $\sigma$ out of its first $2^i$ elements, and compute minimum witness for the results with no witness lying in the first $2^{i-1}$ elements.
    
    To find first $2^{i-1}$ elements out of first $2^i$ elements, suppose there are $t$ result elements (elements of the convolutions results), we consider only first $2^i$ elements and set $k=2\log(t)=\tilde{O}(1)$ in Lemma \ref{lemma:minwitnessf} to find at most $k$ witnesses out of the first $2^i$ elements for every result element.
    
    Consider all result elements with at least $k$ witnesses, we have found $k$ of these witnesses and with Lemma \ref{lemma:hittingset} we can compute a hitting set of size $\le 2^i\log(t)/k=2^{i-1}$. We then set the first $2^{i-1}$ elements to be this hitting set, then these results will all have witness in this first half. We pick any permutation for the remaining $2^{i-1}$ elements in $\sigma$. All results with $<k$ witnesses have all their witnesses computed, and we can compute minimum witness naively for those with no witnesses in the first half.
    
    As a result, we get algorithm \ref{alg:adaptminwitness}.
\end{proof}

\begin{algorithm} [H]
    \caption{Adaptive Minimum Witness} \label{alg:adaptminwitness}
    \begin{algorithmic}[1]
        \Procedure{AdaptiveMinimumWitness}{}
            \State Suppose we have $p$ convolutions of length $n$: $c_i=a_i\otimes b_i$ for $i\in[1,p]$
            \State Let $l=\{(i,j)\mid i\in [1,p],j\in[0,2n-2]\}$
            \For {$m \in [n,n/2,n/4,\cdots,1]$}
                \State $k=2\lceil \log(|l|)\rceil+5$
                \State For each $(p,q)$ in $l$, compute a set of $k$ different witnesses $w[p][q]\subseteq \sigma[0,\cdots,m/2-1]$ for $c[p][q]$ if possible\Comment{Lemma \ref{lemma:minwitnessf}}
                \State Find a hitting set $S\subseteq\sigma[0,\cdots,m-1]$ of size $m/2$ for all $w[p][q]$ of size $\ge k$\Comment{Lemma \ref{lemma:hittingset}}
                \State Permute $\sigma[0,\cdots,m-1]$ so $\sigma[0,\cdots,m/2-1]$ forms $S$
                \For {$(p,q)\in l$}
                    \If{$|w[p,q]|<k~\AND w[p,q]\subseteq \sigma[m/2,\cdots,m-1]$}
                        \State $c[p,q]\gets \min_\sigma w[p,q]$ \Comment{Minimum with respect to order $\sigma$}
                        \State Remove $(p,q)$ from $l$
                    \EndIf
                \EndFor
            \EndFor
        \EndProcedure
    \end{algorithmic}
\end{algorithm}

Combining Theorem \ref{thm:adaptiveminwit}, Lemma \ref{lemma:kernelcompidea}, \ref{lemma:propagate} and \ref{lemma:propagatemod}, we arrive at near-linear solutions for All-Target CoinChange and Residue Table.

\begin{theorem} \label{theo:coin}
    All-Target CoinChange can be solved in $O((u+t)\polylog(u))$ time deterministically.
\end{theorem}
\begin{proof}
    Replace witness finding in Algorithm \ref{alg:kernelcompidea} with Algorithm \ref{alg:adaptminwitness}, we can compute lexicographically minimal solutions for all $(2\log_2 u+1)$-kernels. We then propagate with Lemma \ref{lemma:propagate}.
\end{proof}

\begin{theorem} \label{theo:restable}
    Residue Table can be computed in $O(u\polylog(u))$ time deterministically.
\end{theorem}
\begin{proof}
    Replace witness finding in Algorithm \ref{alg:kernelcompidea} with Algorithm \ref{alg:adaptminwitness} and propagate with Lemma \ref{lemma:propagatemod}.
\end{proof}

\section{Conclusion}

We presented new combinatorial insights and near-optimal algorithms to three generalizations of Unbounded SubsetSum. Our insights and techniques are of independent interest and can also apply to other generalizations for Unbounded SubsetSum. 

\section*{Acknowledgment}

\bibliographystyle{alphaurl}
\bibliography{main}

\newpage
\appendix


\end{document}